\documentclass[12pt]{article}

\usepackage{graphicx, mathtools, amssymb, latexsym, amsmath, amsfonts, amsthm, bbm, tikz}
\usepackage[toc,page]{appendix}

\usepackage{xcolor}
\definecolor{DarkGreen}{rgb}{0.1,0.5,0.1}
\definecolor{DarkRed}{rgb}{0.5,0.1,0.1}
\definecolor{DarkBlue}{rgb}{0.1,0.1,0.5}

\usepackage[small]{caption}
\usepackage[pdftex]{hyperref}
\hypersetup{
    unicode=false,          % non-Latin characters in Acrobat¿s bookmarks
    pdftoolbar=true,        % show Acrobat toolbar?
    pdfmenubar=true,        % show Acrobat menu?
    pdffitwindow=false,      % page fit to window when opened
    pdfnewwindow=true,      % links in new window
    colorlinks=true,       % false: boxed links; true: colored links
    linkcolor=DarkBlue,          % color of internal links
    citecolor=DarkGreen,        % color of links to bibliography
    filecolor=DarkGreen,      % color of file links
    urlcolor=DarkBlue,          % color of external links
    pdftitle={},
    pdfauthor={},    
}

\usepackage{array, fancyhdr, bm, listings, soul, xcolor,titlesec, cancel,ifthen}

\mathtoolsset{showonlyrefs}

\usepackage{hyperref}
\hypersetup{
    colorlinks=true, % make the links colored
    linkcolor=blue, % color TOC links in blue
    urlcolor=red, % color URLs in red
    linktoc=all % 'all' will create links for everything in the TOC
}

\theoremstyle{plain}
\newtheorem{theorem}{Theorem}[section]
\newtheorem{lemma}[theorem]{Lemma}

\newtheorem{proposition}[theorem]{Proposition}

\newtheorem{corollary}[theorem]{Corollary}

\newtheorem*{theorem*}{Theorem}

\theoremstyle{definition}
\newtheorem{definition}[theorem]{Definition}
\newtheorem{remark}[theorem]{Remark}

\newcommand{\floor}[1]{\lfloor {#1} \rfloor}

\newcommand{\ceil}[1]{\lceil {#1}\rceil}

\newcommand{\ind}[1]{^{(#1)}}
\newcommand\defeq{\ensuremath{\stackrel{\rm def}{=}}} % Equal by definition

\newcommand{\poly}{\mathrm{poly}}

\newcommand{\eval}{\textnormal{eval}}
\newcommand{\spn}{\textnormal{span}}
\newcommand{\wt}{\textnormal{wt}}
\newcommand{\cC}{\mathcal{C}}
\newcommand{\F}{\mathbb{F}}
\newcommand\bfa{\textnormal{\textbf{a}}}
\newcommand\bfb{\textnormal{\textbf{b}}}

\newcommand\bfe{\textnormal{\textbf{e}}}

\newcommand\bfi{\textnormal{\textbf{i}}}
\newcommand\bfj{\textnormal{\textbf{j}}}

\newcommand\bfX{\textnormal{\textbf{X}}}

\newcommand\bfZ{\textnormal{\textbf{Z}}}

\begin{document}
\title{Lifted multiplicity codes and the disjoint repair group property\thanks{A conference version of this paper appeared at RANDOM '19.}}
\author{Ray Li\thanks{Department of Computer Science, Stanford University.  Research supported by the National Science Foundation Graduate Research Fellowship Program under Grant No. DGE - 1656518.}\ \ and Mary Wootters\thanks{Departments of Computer Science and Electrical Engineering, Stanford University.  This work is partially supported by NSF grants CCF-1657049 and CCF-1844628.}}
\date{\today}
\maketitle

\begin{abstract}
  Lifted Reed-Solomon Codes (Guo, Kopparty, Sudan 2013) were introduced in the context of locally correctable and testable codes. They are multivariate polynomials whose restriction to any line is a codeword of a Reed-Solomon code. We consider a generalization of their construction, which we call \emph{lifted multiplicity codes}. These are multivariate polynomial codes whose restriction to any line is a codeword of a multiplicity code (Kopparty, Saraf, Yekhanin 2014). We show that lifted multiplicity codes have a better trade-off between redundancy and a notion of locality called the $t$-disjoint-repair-group property than previously known constructions. As a corollary, they also give better tradeoffs for PIR codes in the same parameter regimes. More precisely, we show that, for $t\le \sqrt{N}$, lifted multiplicity codes with length $N$ and redundancy $O(t^{0.585} \sqrt{N})$ have the property that any symbol of a codeword can be reconstructed in $t$ different ways, each using a disjoint subset of the other coordinates.  This gives the best known trade-off for this problem for any super-constant $t < \sqrt{N}$. We also give an alternative analysis of lifted Reed-Solomon codes using dual codes, which may be of independent interest.
\end{abstract}

\section{Introduction}\label{sec:intro}
In this work we study \em lifted multiplicity codes, \em and show how they provide improved constructions of codes with the \em $t$-disjoint repair group property \em ($t$-DRGP), a notion of locality in error correcting codes.

An \em error correcting code \em of length $N$ over an alphabet $\Sigma$ is a set $\cC \subseteq \Sigma^N$.  
There are several desirable properties in error correcting codes, and in this paper we study the trade-off between two of them.  The first is the size of $\cC$, which we would like to be as big as possible given $N$.  
The second desirable property is 
 \em locality. \em  Informally, a code $\cC$ exhibits locality if, given (noisy) access to $c \in \cC$, one can learn the $i$'th symbol $c_i$ of $c$ in sublinear time.  As we discuss more below, locality arises in a number of areas, from distributed storage to complexity theory.

Two constructions of codes with locality are \em lifted codes \em \cite{GuoKS13} and \em multiplicity codes \em \cite{KSY14}; in fact, both of these constructions were among the first known high-rate Locally Correctable Codes.  In this work, we consider a combination of the two ideas in \em lifted multiplicity codes, \em and we show that these codes exhibit locality beyond what's known for either lifted codes or for multiplicity codes.

More precisely,
we study a particular notion of locality called the \em $t$-disjoint-repair-group property \em ($t$-DRGP).  
Informally, we say that $\cC$ has the $t$-DRGP if any symbol $c_i$ of $c \in \cC$ can be obtained in $t$ different ways, each of which involves a disjoint set of coordinates of $c$.  Formally, we have the following definition. 
\begin{definition}
A code $\cC \subseteq \Sigma^N$ has the \em $t$-disjoint repair property \em if for every $i \in [N]$, there is a collection of $t$ disjoint subsets $S_1, \ldots, S_t \subseteq [N] \setminus \{i\}$, and functions $f_1, \ldots, f_t$ so that for all $c \in \cC$ and for all $j \in [t]$, $f_j(c|_{S_j}) = c_i$.  The sets $S_1, \ldots, S_t$ are called \em repair groups. \em 
\end{definition}

As discussed more in Section~\ref{sec:related} below, the $t$-DRGP naturally interpolates between many different notions of locality.
The $t$-DRGP is well-studied both when $t = O(1)$ is small (where it is related to Locally Repairable Codes and nearly equivalently to Private Information Retrieval Codes)  and $t = \Omega(N)$ is large (where it is equivalent to Locally Correctable Codes).
For this reason, it is natural to study the $t$-DRGP when $t$ is intermediate; for example, when $t = N^a$ for $a\in(0,1)$.  In this case, it is possible for the size of the code $|\cC|$ to be quite large: more precisely, it is possible for the \em rate \em  $R = \frac{\log_{|\Sigma|}|\cC|}{N}$ to approach $1$ (notice that we always have $|\cC| \leq |\Sigma|^N$, hence we always have $R\leq 1$).  Thus, the goal is to understand exactly how quickly the rate can approach $1$.  That is, given $t$, how small can the \em redundancy \em $N - RN$ be?

Several works have tackled this question, and we illustrate previous results in Figure~\ref{fig:litreview}.  
Our main result is that lifted multiplicity codes improve on the best-known trade-offs for all super-constant $t \leq \sqrt{N}$. 

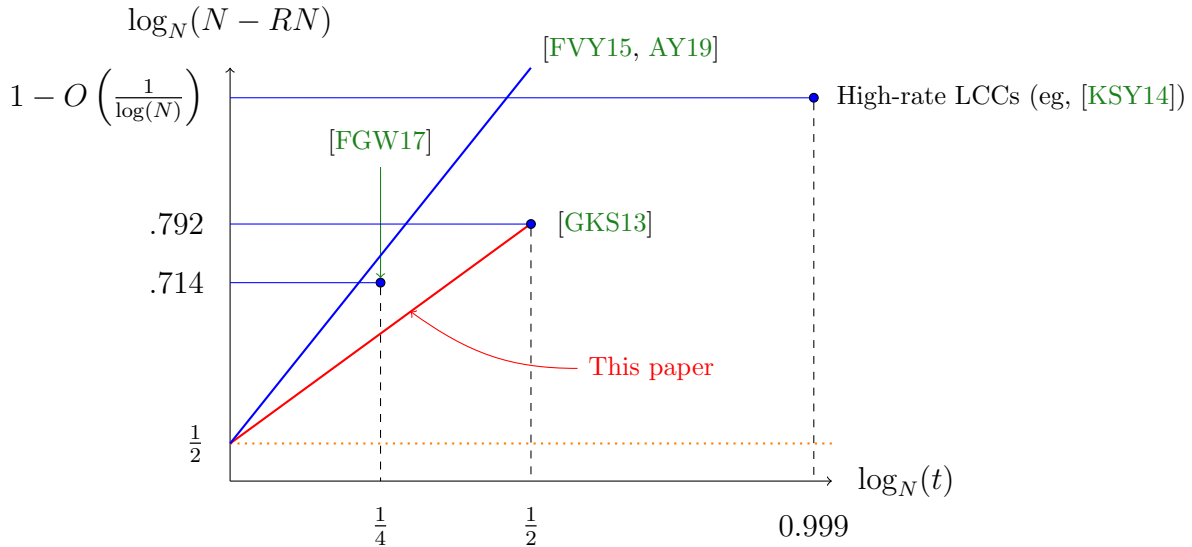
\begin{figure}
\centering
\begin{tikzpicture}[xscale=8, yscale=10]
\draw[->] (0,.45) to (1,.45);
\draw[->] (0,.45) to (0,1);
\node[anchor=west] at (1.025,.45) {$\log_N(t)$};
\node[anchor=south] at (0,1.025) {$\log_N( N - RN )$};
\draw[dotted, thick,orange] (0,.5) to (1, .5);
\node[anchor=east] at (-.025, .5) {$\frac{1}{2}$};
\node[anchor=east] at (-.025, .714) {$.714$};
\node[anchor=east] at (-.025, .792) {$.792$};
\node[anchor=east] at (-.025, .96) {$1 - O\left( \frac{1}{\log(N)}\right)$};
\node[draw, fill=blue, circle,scale=0.3](a) at (.25, .714) {};
\draw[blue] (a) to (0, .714);
\draw[dashed] (a) to (.25, .45);
\node at (.25, .39) {$\frac{1}{4}$};
\node[rotate=0](fgw) at (.25, .9) {\footnotesize \cite{FischerGW17}};
\draw[->,DarkGreen] (fgw) to (a);
\draw[domain=0:.5, smooth, variable=\x, red, thick] plot ( {\x} , {.5 + \x * .585});
\node[draw, fill=blue, circle, scale=0.3](b) at (.5, .792) {};
\draw[blue] (b) to (0, .792);
\node[anchor=west](lrs) at (.525, .792) {\footnotesize \cite{GuoKS13}};
\node[draw,fill=blue,scale=0.3,circle](c) at (.97, .96) {};
\node[anchor=west] at (.99, .96) {\footnotesize High-rate LCCs (eg, \cite{KSY14})};
\draw[blue] (c) to (0, .96);
\draw[dashed] (c) to (.97,.45);
\node at (.97,.39) {$0.999$};
\draw[dashed] (b) to (.5, .45);
\node at (.5, .39) {$\frac{1}{2}$};
\draw[thick,blue, domain=0:.5, smooth, variable=\y] plot ( {\y}, {\y + .5} );
\node[anchor=west](fvy) at (.5,1.025) {\footnotesize \cite{FVY15,AY17}};
\node[red](us) at (.7, .6) {\footnotesize This paper};
\draw[red,->] (us) to [out=180,in=-30] ( .3 , .6755 );

\end{tikzpicture}

\caption{The best trade-offs known between the number $t$ of disjoint repair groups and the redundancy $N - RN$.  Blue points and lines indicate upper bounds (possibility results), and the red line indicates our upper bound.  The best lower bound (impossibility result) available is that we must have $\log_N((1-R)N) \geq 1/2$ for any $t \geq 2$, and this is shown as the dotted orange line.  }\label{fig:litreview}
\end{figure}

\paragraph{Contributions.}
We summarize the main contributions of this work below.
\begin{enumerate}
\item
For $t\le \sqrt{N}$, we construct codes with the $t$-DRGP and redundancy at most  
$$O\left(t^{\log_2(3) - 1} \sqrt{N}\right) \approx O\left(t^{0.585}\sqrt{N}\right).$$  This gives the best known construction for all $t$ with $t = \omega(1)$ and $t < \sqrt{N}$; the only previous result that held non-trivially for a range of $t$ was redundancy $O(t \sqrt{N})$~\cite{FVY15,BE16,AY17} and our result also surpasses the specialized bound for $t=N^{1/4}$ of \cite{FischerGW17}.  

We note, however, that our construction has a large alphabet size, $N^{\Theta(N/t^2)}$. 
In contrast, the works \cite{GuoKS13,FVY15,FischerGW17} have alphabet size at most polynomial in $N$.  However, we can follow the approach of \cite{AY17} and make our code binary by replacing each symbol with an (uncoded) binary string.  This yields \em binary \em codes with $t$-DRGP, that have the best known trade-offs between $t$ and the redundancy when $N^{1/4} < t < N^{1/2}$ among all known codes with alphabet size $\poly(N)$.

\item
We give a new analysis of bivariate lifts of multiplicity codes.  Both multiplicity codes and lifted codes have been studied before (even in the context of the $t$-DRGP), but to the best of our knowledge the only work to consider lifted multiplicity codes is \cite{Wu15}.  That work studies $m$-variate lifts of multiplicity codes, where $m$ is large; its goal is to obtain new constructions of high-rate locally correctable codes.  In the context of our discussion, this corresponds to the $t$-DRGP when $t = N^{0.99}$.  In contrast, for bivariate lifts, we are able to obtain more refined bounds which lead to improved results for the $t$-DRGP when $t \leq \sqrt{N}$.
\end{enumerate}

\paragraph{Organization.} 
In the remainder of the introduction, we survey related work and give an overview of our approach.  In Section~\ref{sec:background}, we give the formal definitions about polynomials and derivatives that we need.  In Section~\ref{sec:construction}, we formally define lifted multiplicity codes.  In Section~\ref{sec:rate}, we prove that lifted multiplicity codes have high rate, and in Section~\ref{sec:drgp}, we prove that they have the $t$-DRGP, which gives rise to our main theorem, Theorem~\ref{thm:main}.

\subsection{Background and Related Work}\label{sec:related}
\subsubsection{Disjoint Repair Groups}
The $t$-DRGP and related notions have been studied both implicitly and explicitly across several communities.  
When $t = O(1)$ is small, several notions related to the $t$-DRGP have been studied, motivated primarily by distributed storage.  These include Locally Repairable Codes (LRCs) with availability~\cite{WZ14,RPDV14,TB_mult_rec,TBF16},  codes for Private Information Retrieval (PIR)~\cite{FVY15,BE16,AY17} (all codes with the $t$-DRGP are $t$-PIR codes) and batch codes~\cite{IKOS04,DGRS14,AY17}; we refer the reader to \cite{Ska16} for a survey of these notions.\footnote{In many (but not all) of these notions, we also care about the size of the repair groups but in this work we focus on the simpler problem of the $t$-DRGP.}

To see why the $t$-DRGP might be relevant for distributed storage, consider a setting where some data is encoded as $c \in \cC$, and then each $c_i$ is sent to a separate server.  If server $i$ is later unavailable, 
we might want to reconstruct $c_i$ without contacting too many other servers.  This can be done if each symbol has one small repair group; this is the defining property of LRCs.
Now suppose that several (say, $t-1$) servers are unavailable.
If $\cC$ has the $t$-DRGP then all $t-1$ unavailable symbols can be locally reconstructed: each node has at least $t$ disjoint repair groups and at most $t-1$ of them have been compromised.  

On the other hand, 
when $t = \Omega(N)$ is large, the $t$-DRGP has been studied in the context of Locally Decodable Codes and Locally Correctable Codes (LDCs/LCCs).  In fact, the $\Omega(N)$-DRGP is equivalent to a constant-query LCC, and the notion has been used to prove impossibility results for such codes~\cite{KT00,Woo10}.

Because of these motivations, there are 
several constructions of $t$-DRGP codes for a wide range of $t$; we illustrate the relevant ones in Figure~\ref{fig:litreview}.  In the context of coded PIR, \cite{FVY15,BE16,AY17} give constructions of $t$-DRGP codes with redundancy $O(t\sqrt{N})$.  This is known to be tight for $t=2$~\cite{RV16, Woo16}, but no better lower bound is known.\footnote{When the size $s$ of the repair groups is bounded, it is known that the redundancy must be at least $\Omega(N \ln(t)/s)$~\cite{TBF16}.}  When $t = \Omega(N)$ is very large, constructing codes with the $t$-DRGP is equivalent to constructing constant-query LCCs, and it is known that the rate of the code must tend to zero~\cite{Woo10}.  On the other hand, for any $\epsilon > 0$, when $t = O(N^{1 - \epsilon})$ is just slightly smaller, then work on high-rate LCCs~\cite{KSY14, GuoKS13, HOW13, KMRS16} (see also \cite{AY17}) imply that there are codes with rate $0.99$ (or any constant less than $1$) with the $t$-DRGP.\footnote{In fact we may even take $\epsilon$ slightly sub-constant using the construction of \cite{KMRS16}.}

When $t = \sqrt{N}$, there are a few constructions known that beat the $O(t \sqrt{N})$ bound mentioned above, including difference-set codes (see, e.g.,~\cite{LinCostello}) and, relevant for us, lifted parity-check codes~\cite{GuoKS13}.  These constructions achieve redundancy $N^{\log_4(3)} \approx N^{0.79}$ when $t = \sqrt{N}$.  In Appendix~\ref{app:dual}, we include a new proof of the fact that the lifted codes of \cite{GuoKS13} have this redundancy using a dual view of lifted codes. 

When $t < \sqrt{N}$, there is only one construction known which beats the $O(t \sqrt{N})$ bound, due to \cite{FischerGW17}.  For the special case of $t = N^{1/4}$, they give a construction based on ``partially lifted codes'' which has redundancy $O(N^{0.72}) = O(t^{0.88} \sqrt{N})$.

\subsubsection{Lifting and multiplicity codes}

Lifted multiplicity codes are based on lifted codes and multiplicity codes, both of which have a long history in the study of locality in error correcting codes.

\paragraph{Lifted Codes.}
Lifting was introduced by Guo, Kopparty and Sudan in \cite{GuoKS13}.  The basic idea can be illustrated by Reed-Solomon (RS) codes.  An RS code of degree $d$ over $\F_q$ is the code 
\[ \mathsf{RS}_{d,q} = \left\{ (f(x_1), \ldots, f(x_q)) \,:\, f \in \F_q[X], \deg(f) < d \right\}, \]
where $x_1, \ldots, x_q$ are the elements of $\F_q$.  There is a natural multi-variate version of RS codes, known as Reed-Muller codes:
\[ \mathsf{RM}_{d,q,m} =
\left\{ (f(\mathbf{x}_1), \ldots, f(\mathbf{x}_{q^m})) \,:\, f \in \F_q[X_1, \ldots, X_m], \deg(f) < d \right\}, \]
where $\mathbf{x}_1, \ldots, \mathbf{x}_{q^m}$ are the elements of $\F_q^m$.  Reed-Muller codes have a very nice locality property, which is that the restriction of a RM codeword to a line in $\F_q^m$ yields an RS codeword.  This fact has been taken advantage of extensively in applications like local decoding, local list-decoding and property testing.
However, RM codes have a downside, which is that if $d < q$ (required for the above property to kick in), they have very low rate.  With this inspiration, we could ask for the set $\cC$ which contains evaluations of \em all \em $m$-variate polynomials which restrict to low-degree univariate polynomials on every line.  Surprisingly, \cite{GuoKS13} showed that this set $\cC$ can be much larger than the corresponding RM code!  This code $\cC$ is called a \em lifted \em Reed-Solomon code, 
and the main structural result of \cite{GuoKS13} is that $\cC$ is the span of the monomials whose restrictions to lines are low-degree.  This property is key when analyzing the rate of these codes.   Moreover \cite{GuoKS13} showed that this is the case when we begin with \em any \em affine-invariant code, not just RS codes.

The original motivation for lifted codes was to construct LCCs, but \cite{GuoKS13} actually also give a code with the $\sqrt{N}$-DRGP, mentioned above; we give an alternate proof that this construction has the $\sqrt{N}$-DRGP in Appendix~\ref{app:dual}.  A variant of lifting was also used in \cite{FischerGW17} to construct $N^{1/4}$-DRGP codes; however, the analysis of this construction is quite brittle and seems difficult to extend to non-trivial constructions for $t \neq N^{1/4}$. 

\paragraph{Multiplicity Codes.}
Multiplicity codes were introduced by Kopparty, Saraf and Yekhanin~\cite{KSY14} with the goal of constructing high-rate LCCs.  The basic idea of multiplicity codes is to get around the low rate of RM codes discussed above in a different way, by appending derivative information to allow for higher-degree polynomials.  That is, it is not useful to have an RS code with degree $d > q$, since $x^q = x$ for any $x \in \F_q$.  However, if we replace the single evaluation $f(x)$ with a vector of evaluations $(f(x), f^{(1)}(x), \ldots, f^{(r-1)}(x))$, where $f^{(i)}$ denotes the $i$'th derivative, then it does make sense to take $d > q$.  The $m$-variate multiplicity code $\mathsf{Mult}_{d,q,m,r}$ of degree $d$ and order $r$ over $\F_q$ is then defined similarly to $\mathsf{RM}_{d,q,m}$:
\[ \mathsf{Mult}_{d,q,m,r} =
\left\{ (f^{(<r)}(\mathbf{x}_1), \ldots, f^{(<r)}(\mathbf{x}_{q^m})) \,:\, f \in \F_q[X_1, \ldots, X_m], \deg(f) < d \right\}, \]
where $f^{(<r)}(\mathbf{x}) \in \F_q^{{m + r - 1 \choose m}}$ is a vector containing all of the partial derivatives of $f$ of order less than $r$, evaluated at $\mathbf{x}$.
Since their introduction,
multiplicity codes have found several uses beyond LCCs, including list-decoding~\cite{Kop15, GW13}, and have even been used to explicitly construct codes with the $t$-DRGP~\cite{AY17}.

\paragraph{Lifted Multiplicity Codes.}
To the best of our knowledge, the only work to study lifted multiplicity codes is the work of Wu~\cite{Wu15}.  The goal of that work is to obtain versions of multiplicity codes which are still high-rate LCCs but which require lower-order derivatives than the construction of \cite{KSY14}.  
The main result in \cite{Wu15} is that lifted multiplicity codes of rate $1 - \alpha$ are LCCs with locality $N^{\epsilon}$ (this corresponds roughly to having the $t$-DRGP with $t = O(N^{1 - \epsilon})$).  However, since the number of variables in the lift is large, it is hard to get a very precise handle on the codimension.

In comparison, in our work, we focus on the $t$-DRGP for $t \leq \sqrt{N}$, but where our goal is to get much tighter bound on the codimension of the code.  We address the quantitative comparison between our bound on the rate and that obtainable by the techniques of~\cite{Wu15} in Remark~\ref{rem:Wu15}.

We note that the construction in \cite{Wu15} is similar to the construction presented here.
Since this construction is somewhat non-trivial (for reasons discussed below), we include the details. 

\paragraph{Why only bivariate lifts?}
In contrast to \cite{Wu15}, we study \em bivariate \em lifts of multiplicity codes.
By focusing only on bivariate lifts (as was also done in \cite{FischerGW17}), we obtain a more precise handle on the codimension of lifted multiplicity codes, which gives results for the $t$-DRGP for $t \leq \sqrt{N}$.  (See Remark~\ref{rem:bivariate} for more on why bivariate lifts make it much easier to analyze the codimension.) 
We believe that this wide range of $t$ is interesting, and thus we think that bivariate lifts are worth focusing on.

We expect that lifted multiplicity codes can be analyzed over more variables.
However, we expect that this will not improve the tradeoff between the redundancy and $t$ (the number of repair groups) for the setting $t\le \sqrt{N}$.
Indeed, this tradeoff becomes worse for ordinary multiplicity codes \cite{AY17}: for these codes, a larger number of variables yields better bounds only for larger values of $t$.
In general, $m$-variable lifted multiplicity codes can have up to $q^{m-1} = N^{(m-1)/m}$ disjoint repair groups, so $\ceil{1/\varepsilon}$ variables are needed for $N^{1-\varepsilon}$ repair groups.
For $N^{(m-2)/(m-1)} \le t\le N^{(m-1)/m}$, we expect that the number of variables that gives the best rate for lifted multiplicity codes is $m$.
We leave the analysis for more variables $m$ this for future work (see Section~\ref{sec:conclude}).

\subsection{Our approach}
We study lifted multiplicity codes to obtain improved constructions of codes with the $t$-DRGP.
We focus on bivariate lifts in this paper in order to obtain codes with $t$-DRGP for $t\le \sqrt{N}$.
We expect that lifted multiplicity codes in more than two variables also give better codes for the $t$-DRGP when $t > \sqrt{N}$.

\subsubsection{Definition of lifted multiplicity codes}
It is not immediately obvious how to apply lifting (and in particular, the nice characterization of it developed in \cite{GuoKS13} as the span of ``good'' monomials) to univariate multiplicity codes.  
We first note that the univariate multiplicity code $\mathsf{Mult}_{d,q,1,r} \subseteq \left(\F_q^r\right)^q$ does not fit the affine-invariant framework of \cite{GuoKS13}, so their results do not immediately apply.  
Instead, we might try to define the bivariate lift of $\mathsf{Mult}_{d,q,1,r}$ as the set
of vectors $(f^{(<r)}(\mathbf{x}_1), \ldots, f^{(<r)}(\mathbf{x}_{q^2}))$ for all polynomials $f$ so that every restriction of $f$ to a line agrees with some polynomial of degree less than $d$ on its first $r-1$ derivatives; that is, the restriction of $f$ is \em equivalent up to order $r$ \em to a polynomial of degree less than $d$.  This works, but there are two non-trivial things to deal with.
\begin{enumerate}
\item First, in order to get a handle on the rate of the code, as in \cite{GuoKS13} we show that the set of valid polynomials $f$ includes the span of a large set of ``good'' monomials. In contrast to \cite{GuoKS13},  the good monomials in this work do not span the entire code. However, lower bounding the number of good monomials, which in turns gives a lower bound on the rate of the code, turns out to be enough for our results.

\item Second, we need to take some care about what monomials we allow.  With lifted RS codes, one only allows monomials $X^a Y^b$ with individual degrees $a,b < q$; otherwise, we could have multiple monomials which correspond to the same codeword which leads to problems if we are counting monomials in order to understand the dimension of the code.  As we show in Lemma~\ref{lem:eval}, it turns out that with multiplicity codes, we should only allow monomials $X^a Y^b$ with $\lfloor a/q \rfloor + \lfloor b/q \rfloor < r$; otherwise, we would have multiple monomials the correspond to the same codeword and this would create similar problems.
\end{enumerate}

Dealing with these issues leads us to the final code and rate analysis, where we define the lifted multiplicity code to be all polynomials spanned by monomials $X^aY^b$ with $\lfloor a/q \rfloor + \lfloor b/q \rfloor < r$, such that the restriction of the polynomial to a line is equivalent up to order $r$ to some univariate polynomial of degree less than $d$.
We then lower bound the number of evaluations of monomials in this code, giving a lower bound on the rate.
We note that the work \cite{Wu15} considers a similar construction.

\subsubsection{Lifted multiplicity codes have the $t$-DRGP}
In Corollary~\ref{cor:mult-6} we give a lower bound on the number of $(q,r,d)$-good monomials, and this leads to a lower bound on the dimension of the lifted multiplicity code; crucially, this can be quite a bit bigger than the dimension of the corresponding multivariate multiplicity code.

Finally, we observe that lifted multiplicity codes have the $t$-DRGP for a range of values of $t$.  Similarly to previous constructions based on multivariate polynomial codes, the disjoint repair groups to recover the symbol $f^{(<r)}(\mathbf{x})$ are given by  disjoint collections of lines through $\mathbf{x}$.  More precisely, the values $f^{(<r)}(\mathbf{y})$ for the set of $\mathbf{y}$ that lie on $r$ distinct lines through $\mathbf{x}$ can be used to recover $f^{(<r)}(\mathbf{x})$.  Thus, the number of disjoint repair groups is $q/r = \sqrt{N}/r$.  By adjusting $r$, we obtain the trade-off shown in Figure~\ref{fig:litreview}.  
Our main theorem is as follows.

\begin{theorem}
  \label{thm:main}
  For $q=2^\ell$ and $r=2^{\ell'}$ with $1\le \ell'\le \ell$, 
there exists a code $\mathcal{C}$  over $\mathbb{F}^{\binom{r+1}{2}}_q$ with the following properties. 
  \begin{itemize}
  \item The length of the code is $q^2$.
  \item The rate of the code is at least
\[ 1 - \frac{ 3r^{\log_2(8/3)}q^{\log_2(3)} }{ \binom{r+1}{2} q^2 }, \]
so that the redundancy is at most
\[ \frac{ 3r^{\log_2(8/3)}q^{\log_2(3)} }{ \binom{r+1}{2}}. \]
  \item The code has the $q/r$-disjoint repair group property.
  \end{itemize}
\end{theorem}
As a remark, our techniques can also recover any symbol from any one of its repair groups in polynomial time.
For any $\gamma\in[0,1]$, choosing $q=2^\ell$ and $r=2^{\ell'}$ with $\gamma \approx \ell'/\ell$ gives a code with length
$N = q^2$ and redundancy at most
\[ 6N^{\log_4(3)- \gamma (1- \log_4(8/3)) } \]
with the $N^{(1 - \gamma)/2}$-DRGP.  This is made formal in the following corollary.
 
\begin{corollary}
  \label{cor:main}
For any $\epsilon \in(0,\frac{1}{2})$, there are infinitely many $N$ so that, for $t=\left\lfloor N^\epsilon \right\rfloor $,
there exists a code of length $N$ which has the $t$-DRGP and redundancy at most 
$6 t^{\log_2(3) - 1} \sqrt{N}. $
\end{corollary}

We note that Theorem~\ref{thm:main} also yields results for constant $t$, not just for $t = N^{\epsilon}$ as presented in Corollary~\ref{cor:main}.  For example, by setting $r=q/2$ we
obtain a code with the $2$-DRGP and redundancy at most $9\sqrt{N}$.  The constant $9$ is not optimal here (the optimal constant for $t=2$ is known to be $\sqrt{2}$~\cite{RV16}), but to the best of our knowledge, Theorem~\ref{thm:main} does yield the best known bounds for any super-constant $t$.

The codes in Theorem~\ref{thm:main} and Corollary~\ref{cor:main} have the disadvantage of having a large alphabet size.
Indeed, we have $r = q/t$, and so the alphabet size is $Q=q^{\binom{r+1}{2}} = N^{\Theta(N/t^2)}$, which is very large. 
It is an interesting question to obtain the results of Corollary~\ref{cor:main} with a code over a smaller alphabet (see open questions in Section~\ref{sec:conclude}).
Among the existing work in Figure~\ref{fig:litreview}, \cite{GuoKS13, FVY15, FischerGW17} all have $\poly(N)$ or smaller sized alphabets.

For now, we observe as in \cite{AY17} that, if $\mathcal{C}$ is a code with the $t$-DRGP, then replacing $\mathcal{C}$ with a binary code $\mathcal{C}'$, where each symbol in each codeword is replaced with $\log(Q)$ binary bits, yields a code that also has the $t$-DRGP.
As a result, applying this to the code in Corollary~\ref{cor:main} yields a code with length $N_{bin} = N\log(Q) = N^{2-2\varepsilon}$ and redundancy $O(t^{\log_2(3/2)}\sqrt{N}\log(Q))=O(N^{3/2+\varepsilon\log_2(3/8)}\log N)$.
\begin{corollary}
  \label{cor:binary}
For any $\epsilon \in(0,\frac{1}{2})$, there are infinitely many $N$ so that, for $t=\left\lfloor N^\frac{\epsilon}{2-2\varepsilon} \right\rfloor $,
there exists a \emph{binary} code of length $N$ which has the $t$-DRGP and redundancy at most 
$\tilde O(N^{\frac{3/2 + \varepsilon\log_2(3/8)}{2-2\varepsilon}})$.
\end{corollary}
Among codes with alphabet size $\poly(N)$ or smaller, our binary codes give the best known tradeoff between $t$ and redundancy when $N^{1/4}<t<N^{1/2}$ 
(at $t=N^{1/4}$ \cite{FischerGW17} gives a better redundancy).

\section{Preliminaries}\label{sec:background}
In this section, we introduce the background we need on polynomials and derivatives over finite fields.
Throughout this paper, we assume that $q$ is a power of 2.
Let $\mathbb{F}_q$ denote the finite field of order $q$, and let $\mathbb{F}_q^*$ denote its multiplicative subgroup.

If $a$ and $b$ are nonnegative integers with binary representations $a=\overline{a_{\ell-1}\cdots a_0}$ and $b=\overline{b_{\ell-1}\cdots b_0}$, then we write $a\le_2 b$ if $a_i\le b_i$ for $i=0,\dots,\ell-1$.
If $a$ is an integer, let $(a\mod c)$ denote the element of $\{0,\dots,c-1\}$ congruent to $a$ mod $c$.
We write $a\le_2^\ell b$ if $(a\mod 2^\ell)\le_2 (b\mod 2^\ell)$.

As in \cite{GuoKS13}, we use Lucas's theorem.
\begin{proposition}[Lucas's theorem]
\label{prop:lucas}
  Let $p$ be a prime and $a=\overline{a_{\ell-1}\cdots a_0}, b=\overline{b_{\ell-1}\cdots b_0}$ be written in base $p$.
  Then
  \begin{align}
    \binom{a}{b}\equiv\prod_{i=0}^{\ell-1} \binom{a_i}{b_i}\mod p
  \label{}
  \end{align}
  In particular, if $p=2$, then $\binom{a}{b}\equiv 1\mod p$ if and only if $b\le_2 a$.
\end{proposition}

\subsection{Polynomials and derivatives}
For a vector $\bfi = (i_1,\dots,i_m)$ of nonnegative integers, its \emph{weight}, denoted $\wt(\bfi)$, equals $\sum_{k=1}^{m} i_k$.
For a field $\mathbb{F}$, let $\mathbb{F}[X_1,\dots,X_m]=\mathbb{F}[\bfX]$ be the ring of polynomials in the variables $X_1,\dots,X_m$ with coefficients in $\mathbb{F}$.
For a vector of nonnegative integers $\bfi=(i_1,\dots,i_m)$ and a vector $\bfX=(X_1,\dots,X_m)$ of variables, let $\bfX^\bfi$ denote the monomial $\prod_{j=1}^{m} X_j^{i_j}\in\mathbb{F}[\bfX]$, and for a vector $\bfa=(\alpha_1,\dots,\alpha_m)\in\mathbb{F}^m$, let $\bfa^\bfi$ denote the value $\prod_{j=1}^{m} \alpha_j^{i_j}$, where $0^0\defeq 1$.
For nonnegative vectors $\bfi=(i_1,\dots,i_m)$ and $\bfj=(j_1,\dots,j_m)$, we write $\bfi\le \bfj$ if $i_k\le j_k$ for all $k$.
We also write $\binom{\bfi+\bfj}{\bfi}$ to denote $\prod_{k=1}^{m} \binom{i_k+j_k}{i_k}$. 
For nonnegative vector $\bfi$, we let $[\bfX^\bfi]P(\bfX)$ denote the coefficient of $\bfX^\bfi$ in the polynomial $P(\bfX)$.

We will use Hasse derivatives, a notion of derivatives over finite fields:
\begin{definition}[Hasse derivatives]
  \label{def:deriv}
  For $P(\bfX)\in \mathbb{F}[\bfX]$ and a nonnegative vector $\bfi$, the $\bfi$-th (Hasse) derivative of $P$, denoted $P\ind{\bfi}(\bfX)$ or $D\ind{\bfi}P(\bfX)$, is the coefficient of $\bfZ^\bfi$ in the polynomial $\tilde P(\bfX,\bfZ)\defeq P(\bfX+\bfZ) \in\mathbb{F}[\bfX,\bfZ]$.
  Thus,
  \begin{align}
    P(\bfX+\bfZ) = \sum_{\bfi}^{} P\ind{\bfi}(\bfX)\bfZ^i.
  \end{align}
\end{definition}
For $\mathbf{x} \in \F_q^m$ and $P(X)\in \mathbb{F}_q[\bfX]$, 
we use the notation $P^{(<r)}(\mathbf{x}) \in \F_q^{{m + r - 1\choose m}}$
to denote the vector containing $P^{(\mathbf{i})}(\mathbf{x})$ for all $\mathbf{i}$ so that $\wt(\mathbf{i}) < r$.   We record a few useful (well-known) properties of Hasse derivatives below (see \cite{HKT08}). 
\begin{proposition}[Properties of Hasse derivatives]
  \label{prop:deriv-1}
  Let $P(\bfX),Q(\bfX)\in\mathbb{F}[\bfX]$ and let $\bfi,\bfj$ be vectors of nonnegative integers. Then
  \begin{enumerate}
  \item $P\ind{\bfi}(\bfX) + Q\ind{\bfi}(\bfX) = (P+Q)\ind{\bfi}(\bfX)$.
  \item $(P\cdot Q)\ind{\bfi}(\bfX) = \sum_{0\le\bfe\le\bfi}^{} P\ind{\bfe}(\bfX)\cdot Q\ind{\bfi-\bfe}(\bfX)$.
  \item $(P\ind{\bfi})\ind{\bfj}(\bfX) = \binom{\bfi + \bfj}{\bfi}P\ind{\bfi+\bfj}(\bfX)$.
  \end{enumerate}
\end{proposition}
Using the above, we obtain the following useful derivative computation, and we provide a proof in Appendix~\ref{app:A} for completeness.
\begin{proposition}
  \label{prop:deriv-2}
  Let $1\le r<q$ with $q$ a power of 2, and let $P(X)=(X^q-X)^r$. Then, 
  \begin{align}
    P\ind{i}(X) = \left\{
    \begin{matrix}
    \binom{r}{i}(X^q-X)^{r-i} & 0\le i\le r\\
    0 & i > r\\
    \end{matrix}
    \right.
  \end{align}
\end{proposition}

\subsection{Polynomial local recovery}
A key property exploited by earlier work on multiplicity codes~\cite{KSY14,Kopparty15} is that $f^{(<r)}(\mathbf{x})$ can be recovered from $f^{(<q)}(\mathbf{y})$ for $\mathbf{y}$ that lie on a collection of lines through $\mathbf{x}$. More precisely,
let $\mathcal{L}_m$ be the set of lines $L(T)$ of the form $\bfa T + \bfb$ with $\bfa,\bfb\in\mathbb{F}_q^m$.
Given a multivariate polynomial $P(\bfX)\in\mathbb{F}_q[X_1,\dots,X_m]$, if $L$ is the line $\bfa T + \bfb$, let $P_L(T)\in\mathbb{F}_q[T]$ denote the univariate polynomial $P(\bfa T + \bfb)$.
Let $\mathcal{L}$ be the set of lines in $\mathbb{F}_q^2$ of the form $L(T) = (T, \alpha T + \beta)$ for $\alpha, \beta \in \mathbb{F}_q$.

For simplicity---and because it is enough for our application to the $t$-DRGP---we will consider only bivariate polynomials in this paper, although (see for example \cite{Kopparty15}) the same basic idea works for any $m$.  We will further specialize to lines in $\mathcal{L}$---that is, lines of the form $L(T) = (T, \alpha T + \beta)$---because it will simplify some computations later in the paper.
With these restrictions,
we can specialize Equation (4) of \cite{Kopparty15} to obtain the following relationship between the derivatives of $P_L(T)$ and the derivatives of $P(X,Y)$.
\begin{lemma}[Follows from, e.g., \cite{KSY14, Kopparty15}]
  \label{lem:local-1}
  Suppose that $L_1,\dots,L_r$ are $r$ lines in $\mathcal{L}$ all passing through a point $(\gamma,\delta)$, with $L_i$ being the line $(T,\alpha_i T+\beta_i)$.
  Then, for all polynomials $P(X,Y)\in\mathbb{F}_q[X,Y]$, the following matrix equality holds for all $i=0,\dots,r-1$.
  \begin{align}
    \label{eq:local}
    \begin{bmatrix}
      P\ind{i}_{L_1}(\gamma) \\
      P\ind{i}_{L_2}(\gamma) \\
      \vdots \\
      P\ind{i}_{L_{i+1}}(\gamma)
    \end{bmatrix}
    \ &= \
    \begin{bmatrix}
      \alpha_1^0 & \alpha_1^1 & \cdots & \alpha_1^{i} \\
      \alpha_2^0 & \alpha_2^1 & \cdots & \alpha_2^{i} \\
      \vdots & \vdots & \ddots & \vdots \\
      \alpha_{i+1}^0 & \alpha_{i+1}^1 & \cdots & \alpha_{i+1}^{i} \\
    \end{bmatrix}
    \begin{bmatrix}
      P\ind{i,0}(\gamma,\delta) \\
      P\ind{i-1,1}(\gamma,\delta) \\
      \vdots \\
      P\ind{0,i}(\gamma,\delta)
    \end{bmatrix}.
  \end{align}
\end{lemma}

When lines $L_1, \ldots, L_i$ are distinct, 
the middle matrix in \eqref{eq:local} is a Vandermonde matrix, and Vandermonde matrices are invertible in polynomial time.
Hence, we immediately have the following corollary.

\begin{corollary}
  \label{lem:local-2}
  Suppose that $L_1,\dots,L_r$ are $r$ distinct lines of the form $L_k(T) = (T, \alpha_k T + \beta_k)$ 
all passing through a point $(\gamma,\delta)\in\mathbb{F}_q^2$.
  For a polynomial $P(X,Y)\in\mathbb{F}_q[X,Y]$, given the polynomials $P_{L_1}(T),\dots,P_{L_k}(T)$, the derivatives $P\ind{\bfi}(\gamma,\delta)$ are uniquely determined and computable efficiently for all $\bfi$ such that $\wt(\bfi) < r$.
\end{corollary}

\section{Lifted multiplicity codes}\label{sec:construction}

In this section, we define lifted multiplicity codes.  As noted in the introduction, we restrict our attention to bivariate codes because this is enough for our application to the $t$-DRGP.  However, everything in this section extends to general $m$-variate codes.  We define bivariate lifted multiplicity codes as the vectors $(f^{(<r)}(\mathbf{x}))_{\mathbf{x} \in \mathbb{F}_q^2}$ for polynomials $f(X)$ that live in the span of ``good'' monomials.  In order to define these ``good'' monomials, we need a few more definitions.

\subsection{Polynomial equivalence}
We first define a notion of polynomial equivalence.

\begin{definition}
  We say that two univariate polynomials $A(X),B(X)\in\mathbb{F}_q[X]$ are \emph{equivalent up to order $r$}, written $A\equiv_r B$, if $A\ind{i}(\gamma)=B\ind{i}(\gamma)$ for all $i=0,\dots,r-1$ and $\gamma\in\mathbb{F}_q$.
\end{definition}
It is easy to see that the above definition does in fact give an equivalence relation.
We now present two standard results regarding this equivalence relation.
The first is a characterization of this equivalence.
\begin{lemma}
  \label{lem:equiv}
  For $A(X), B(X)\in\mathbb{F}_q[X]$ we have $A(X)\equiv_rB(X)$ if and only if $(X^q-X)^r|A(X)-B(X)$.
\end{lemma}
\begin{proof}
  By considering the polynomial $A(X)-B(X)$, it suffices to prove $A(X)$ is equivalent to the zero polynomial up to order $r$ if and only if $(X^q-X)^r|A(X)$.
  If $A(X) = (X^q-X)^rC(X)$ for some polynomial $C(X)\in\mathbb{F}_q[X]$, then, by part 2 of Proposition~\ref{prop:deriv-1} and Proposition~\ref{prop:deriv-2}, for $0\le i < r$, we have $X^q-X|A\ind{i}(X)$, so $A\ind{i}(\gamma)=0$ for all $0\le i<r$ and all $\gamma\in \mathbb{F}_q$, so $A(X)\equiv_r0$.
  
  Conversely, suppose that $A(X)\equiv_r 0$.
  By the definition of Hasse derivatives, we have $A(X) = A(\gamma + (X-\gamma)) = \sum_{i}^{} A\ind{i}(\gamma)(X-\gamma)^i$.
  Since $A\ind{i}(\gamma)=0$ for $i=0,\dots,r-1$, we have $(X-\gamma)^r|A(X)$.
  Thus is true for all $\gamma$, so $\prod_{\gamma}^{} (X-\gamma)^r|A(X)$, so $(X^q-X)^r|A(X)$.
\end{proof}
Lemma~\ref{lem:equiv} gives the following corollary.
\begin{lemma}
 Let $q$ be a power of 2 and $r\ge 1$.
 For every univariate polynomial $A(X)$, there exists a unique degree-at-most $rq-1$ polynomial $B(X)$ such that $A(X)\equiv_rB(X)$.
 Furthermore, if $r$ is a power of 2, then for all $a$ such that $\deg A - (qr-r) < a < qr$, we have $[X^a]A(X) = [X^a]B(X)$.
\label{lem:mult-2}
\end{lemma}
\begin{proof}
  For existence of $B(X)$, note that, by Lemma~\ref{lem:equiv}, we can take $B(X)$ to be the remainder when $A(X)$ is divided by $(X^q-X)^r$.
  For uniqueness of $B(X)$, suppose that $B_1(X)$ and $B_2(X)$ are equivalent to $A(X)$ up to order $r$ and are of degree at most $rq-1$.  
  By Lemma~\ref{lem:equiv}, we have $(X^q-X)^r|B_1(X)-B_2(X)$.
  Additionally, $B_1(X)-B_2(X)$ has degree at most $rq-1$, so $B_1(X)-B_2(X)=0$.

  Now suppose $r$ is a power of 2.
  Then $(X^q-X)^r = X^{rq}+X^r$.
  Above, to obtain $B(X)$ from $A(X)$, we need only to subtract terms of the form $X^{qr}+X^r, X^{qr+1}+X^{r+1},\dots,X^{\deg A}+X^{\deg A - qr+r}$.
  Thus, for $a$ such that $\deg A - qr+r < a < qr$, the coefficients of $X^a$ in $A(X)$ and $B(X)$ are equal.
\end{proof}

\subsection{Type-$r$ polynomials}

Define
  the order-$r$ evaluation map $\eval_{q,r}:\mathbb{F}_q[X,Y]\to \left(\mathbb{F}_q^{\binom{r+1}{2}}\right)^{q^2}$ by
  \begin{align}
    \eval_{q,r}(P) := (P^{(<r)}(\mathbf{x}))_{\mathbf{x} \in \F_q^2},
  \end{align}

We will want to restrict our attention to a subset of monomials $M(X,Y) = X^aY^b$ whose order-$r$ evaluations $\eval_{q,r}(M)$ form a basis for the space $\{\eval_{q,r}(P) \,:\, P \in \F_q[X,Y] \}$.  To that end, we introduce the following definition.

\begin{definition}[Type-$r$ monomials]
  Call a monomial $X^aY^b$ \emph{type-$r$} if $\floor{a/q}+\floor{b/q}\le r-1$.
Let $\mathcal{F}_{q,r}$ be the family of polynomials $P\in\mathbb{F}_q[X,Y]$ that are spanned by type-$r$ monomials.
\end{definition}
It is easy to see that $\mathcal{F}_{q,r}$ is a dimension $\binom{r+1}{2}q^2$ vector space over $\mathbb{F}_q$.
We now show that the type-$r$ polynomials form a basis for bivariate polynomials, up to order $r$ equivalence.
We note that Lemma III.1 of \cite{Wu15} claims a similar statement, with a different argument.
\begin{lemma}\label{lem:eval}
  The evaluation map $\eval_{q,r}:\mathcal{F}_{q,r}\to \left(\mathbb{F}_q^{\binom{r+1}{2}}\right)^{q^2}$
  is a bijection.
\label{lem:mult-1}
\end{lemma}
\begin{proof}[Proof of Lemma~\ref{lem:mult-1}]
  Since $\eval_{q,r}$ is a linear map and $\mathcal{F}_{q,r}$ and $\mathbb{F}_{q}^{\binom{r+1}{2}q^2}$ have the same $\mathbb{F}_q$ dimension, it suffices to prove the map has trivial kernel.
  We prove by induction.

  \textbf{Base Case: $r=1$.} Suppose $P\in\mathcal{F}_{q,1}$ and $\eval_{q,1}(P)$ is the 0-vector. Then $P(X,Y)=0$ for all $X,Y$.
  For any $\delta\in \mathbb{F}_q$, the polynomial $P(X,\delta)\in\mathbb{F}_q[X]$ has degree at most $q-1$ but has $q$ roots, so the polynomial must be 0.
  Hence, $(Y-\delta)|P(X,Y)$ for all $\delta$, so $(Y^q-Y)|P(X,Y)$, which implies $P=0$.
  This proves that $\eval_{q,1}$ has trivial kernel.

  \textbf{Inductive step:} 
  Assume $r\ge 1$ and $\eval_{q,r}$ has trivial kernel.
  We prove that $\eval_{q,r+1}$ has trivial kernel. 

  Assume $P(X,Y)$ is a polynomial spanned by type-$(r+1)$ monomials with all $\bfi$th derivatives equal to 0 for $\wt(\bfi)<r+1$.
  Let $\delta\in\mathbb{F}_q$ and $B_\delta(X)\defeq P(X,\delta)$.
  Then, for $0\le i < r$, we have $B_\delta\ind{i}(\gamma) = P\ind{i,0}(\gamma,\delta) = 0$ for all $\gamma\in\mathbb{F}_q$.
  Hence, for all $\gamma\in\mathbb{F}_q$, we have $(X-\gamma)^r|B_\delta(X)$.
  Hence, $(X^q-X)^r|B_\delta(X)$.
  Since $\deg B_\delta(X) \le \deg_XP(X,Y) < qr$ for all $\delta$, we have $B_\delta(X)=0$.
  Thus, $P(X,\delta)$ is the 0 polynomial for all $\delta$, so $(Y-\delta)|P(X,Y)$ for all $\delta$, so $(Y^q-Y)|P(X,Y)$.
  Hence, we may write $P(X,Y)=(Y^q-Y)Q(X,Y)$ for some polynomial $Q(X,Y)\in \mathbb{F}_q[X,Y]$.

  As polynomial $P$ is type-$(r+1)$, polynomial $Q$ is type-$r$: if $Q$ had a nonzero coefficient for $X^aY^b$ with $\floor{a/q}+\floor{b/q}>r-1$, then the coefficient $X^{a}Y^{b+q}$ is nonzero in $P$, which is a contradiction.
  For all $i,j$ with $i\ge 0,j\ge 1$ and $i+j\le r$, we have
  \begin{align}
    P^{(i,j)}(X,Y) 
    \ &= \   (Y^q-Y)Q\ind{i,j}(X,Y)  - Q\ind{i,j-1}(X,Y).
  \end{align}
  Here we applied part 2 of Proposition~\ref{prop:deriv-1} and the $r=1$ case of Proposition~\ref{prop:deriv-2}.
  At every $X$ and $Y$, the left side is 0 by assumption on $P$ and the right side $Q\ind{i,j-1}(X,Y)$.
  We conclude that $Q\ind{i',j'}$ evaluates to 0 everywhere for every nonnegative $i'$ and $j'$ satisfying $i'+j'\le r-1$.
  Since $Q$ is type-$r$, we have $Q=0$ by the induction hypothesis, so $P=0$.
  This completes the induction, completing the proof.
\end{proof}

\subsection{Definition of lifted multiplicity codes}

Finally we are ready to define lifted multiplicity codes, which we define as the set of evaluations $\eval_{q,r}(P)$ of polynomials whose restrictions to lines\footnote{To simplify calculations, we consider restrictions to lines of the form $L(T) = (T, \alpha T + \beta)$.  That is, we do not include lines of the form $L(T)=(\alpha, T)$.} are equivalent, up to order $r$, to a low degree polynomial:
\begin{definition}[Lifted multiplicity codes, first definition]\label{def:lmc2}
  The \emph{$(q,r,d)$ (bivariate) lifted multiplicity code} is a code $\mathcal{C}$ over alphabet $\Sigma=\mathbb{F}_q^{\binom{r+1}{2}}$ of length $q^2$ given by
\[ \mathcal{C} = \left\{ \eval_{q,r}(P) \,:\, \text{ \begin{minipage}{8cm} \begin{center} $P \in \F_q[X,Y]$ and, for any $L(T)\in\mathcal{L}$, $P(L(T)) \equiv_r Q(T)$ for some $Q \in \F_q[T]$ of degree less than $d$. \end{center} \end{minipage}} \right\} \]
\end{definition}

Definition~\ref{def:lmc2} is natural but difficult to get a handle on directly.
Following the approach of previous work \cite{GuoKS13, FischerGW17}, 
we show that lifted multiplicity code contains the set of vectors $\eval_{q,r}(P)$ for $P$ that lie in the span of a set of ``good'' monomials, which makes it easier to bound the rate.  Informally, a monomial is $(q,r,d)$-good if its restriction along every line is equivalent, up to order $r$, to a polynomial of degree less than $d$.

\begin{definition}[$(q,r,d)$-good monomials]\label{def:good}
  Call a monomial $M_{a,b}(X,Y) = X^aY^b\in\mathbb{F}_q[X,Y]$ \emph{$(q,r,d)$-good} (or simply good, when $r$ and $d$ are understood) if it is type-$r$ and for every line $(T,\alpha T+\beta)\in\mathcal{L}$, the univariate polynomial $M_{a,b}(T,\alpha T+\beta)$ is equivalent, up to order $r$, to polynomial of degree less than $d$, and call it \emph{$(q,r,d)$-bad} otherwise.
\end{definition}

By definition all good monomials lie in our lifted multiplicity code, so to lower bound the rate of the code it suffices to lower bound the number of good monomials.
\begin{lemma}
  \label{lem:good}
  Let $\mathcal{C}$ be the bivariate $(q,r,d)$ lifted multiplicity code.
  Then, for every $(q,r,d)$-good monomial $M(X,Y)$, $\eval_{q,r}(M)\in \mathcal{C}$, and the rate of $\mathcal{C}$ is at least $\frac{\#\text{$(q,r,d)$-good monomials}}{\binom{r+1}{2}q^2}$.
\end{lemma}
\begin{proof}
  The first part follows from the definition of good monomial.
  For the second part, $\mathcal{C}$ is linear and the $\mathbb{F}_q$-span of all good monomials have pairwise distinct evaluations by Lemma~\ref{lem:mult-1}, so $|\mathcal{C}|\ge q^{(\#\text{$(q,r,d)$-good monomials})}$. As $\mathcal{C}$ is a length $q^2$ code over an alphabet of size $|\Sigma|=q^{\binom{r+1}{2}}$, the rate is at least $\frac{\log|\mathcal{C}|}{q^2\log|\Sigma|}=\frac{\#\text{$(q,r,d)$-good monomials}}{\binom{r+1}{2}q^2}$.
\end{proof}

\begin{remark}
\label{rem:err}
A previous version of this paper incorrectly asserted that every codeword of the lifted multiplicity code is spanned by good monomials.
As observed by Nikita Polianskii, this is in fact not true. For example, when $r=2$ and $d=2q-1$, the monomials $X^{2q-2}Y$ and $X^{q-1}Y^{q}$ are not $(q,r,d)$-good as verified by the line $(T,T)$, but their sum $X^{2q-2}Y+X^{q-1}Y^q$ is in the $(q,r,d)$-lifted multiplicity code: the restriction of the sum to a line $(T, \alpha T+\beta)\in\mathcal{L}$ has a $T^{2q-1}$ coefficient of $\alpha + \alpha^q = 0$ and hence has degree strictly less than $d=2q-1$.
\end{remark}

\section{The rate of lifted multiplicity codes}\label{sec:rate}
In this section, we bound the rate (and hence, the redundancy) of lifted multiplicity codes.
Our final result on the rate is Corollary~\ref{cor:mult-6} below, which implies that for $r,q$ and $d$ of an appropriate form, the lifted multiplicity code over order $r$ and degree $d$ over $\F_q$ has rate at least
\[ 1 - \frac{6}{r} \left( r - \frac{d}{q} \right)^{\log_2(4/3)}.\]
In the next section, we choose $d = qr - r$, which will yield a code of rate $1 - \frac{6}{r} \left( \frac{r}{q} \right)^{\log_2(4/3)}$ and will give us Theorem~\ref{thm:main}.

Before we prove this result, we briefly compare our approach to more straightforward ones, and discuss why we are able to do better.

First, we discuss what might be a first strategy building on the analysis of \cite{GuoKS13} for lifted Reed-Solomon codes.  Similarly to that work, we want to show there are few bad monomials.
We can show (after checking some conditions) that a monomial is bad if, restricted to some line, in the resulting univariate polynomial, one of the coefficients of $T^{qr-s},T^{qr-s+1},\dots,T^{qr-1}$ is nonzero.
This corresponds to the analysis of lifted Reed-Solomon codes when $r=s=1$.
For each $s'=1,\dots,s$, similar to the analysis of the lifted Reed-Solomon code, we can bound the number of monomials that could cause the coefficient of $T^{rq-s'}$ to be nonzero by $rq^{\log_2(3)}$.
Using the union bound and
summing these bounds gives a bound $rsq^{\log_2(3)}$ on the number of bad monomials for the lifted multiplicity code.
However, when $r=s$ (the setting we will consider), this gives a rate of $1-q^{\log_2(3/4)}$.  Thus, this yields a code with the same redundancy of $N^{\log_4(3)}$ as the lifted Reed-Solomon code, and we have made no improvement.

In order to do better, the key to our analysis is to observe that monomials that are bad for some $s'$ are likely to be bad for another $s''$, so the union bound is wasteful.
Instead, using some tricks with binary arithmetic (captured in Lemma~\ref{lem:mult-3}), we are able to analyze together all the monomials that make any of the coefficients of $T^{rq-s},\dots,T^{rq-1}$ nonzero, giving a better bound.

Second, we compare our approach to the analysis of \cite{Wu15}, which also studies lifted multiplicity codes, but focuses on a different parameter regime (one where $t$ is much larger).
As described more in Remark~\ref{rem:Wu15}, the approach of \cite{Wu15} does not yield anything better in the parameter regime that we consider ($t \leq \sqrt{N}$) than does the approach described above (or indeed even any better than standard (not lifted) multiplicity codes when $r \gg q^{1/23}$).  
The reason that we are able to do better than the straightforward argument above while the approach of \cite{Wu15} does not is that \cite{Wu15} uses a stricter requirement for a monomial to be good in \cite[Lemma III.3]{Wu15} than we do in our Lemma~\ref{lem:mult-5}.  Thus, the approach of \cite{Wu15} counts a smaller number of good monomials and ends up with a weaker bound on the rate. 

Now, we prove our result.
We begin with a lemma that will be useful.
\begin{lemma}
  Let $s=2^{\ell_s}$ and $q=2^\ell$ with $\ell_s\le \ell$.
  The number of $a_1,b_1\in\{0,1,\dots,q-1\}$ such that at least one of the following is true
  \begin{align}
    q-1-a_1 \ &\le_2^\ell \   b_1 \nonumber\\
    q-2-a_1 \ &\le_2^\ell \   b_1 \nonumber\\
    \vdots \ &\vdots \   \vdots \nonumber\\
    q-s-a_1 \ &\le_2^\ell \   b_1
  \label{eq:mult-5}
  \end{align}
  is at most $2\cdot 3^\ell\cdot \left( 4/3 \right)^{\ell_s}=2\cdot 3^\ell\cdot s^{\log_2(4/3)}$.
\label{lem:mult-3}
\end{lemma}
\begin{proof}
  Suppose we write the numbers $(q-1-a_1\mod q),(q-2-a_1\mod q),\dots,(q-s-a_1\mod q)$ in binary with $\ell$ digits (possibly with leading zeros).
  As these number span $2^{\ell_s}$ consecutive integers mod $q$, when written in this binary form, their most significant $\ell-\ell_s$ coordinates take on at most 2 values.
  Let $a_2 = \floor{\frac{(q-1-a_1\mod q)}{2^{\ell_s}}}$ and $b_2=\floor{\frac{b_1}{2^{\ell_s}}}$ so that $a_2,b_2\in\{0,\dots,2^{\ell-\ell_s}-1\}$, and $a_2$ and $b_2$ are the most significant $\ell-\ell_s$ coordinates of $(q-1-a_1\mod q)$ and $b_1$, respectively, when written in $\ell$-digit binary.
  Then if one of the equations of \eqref{eq:mult-5} is true, then we must have either $a_2\le_2 b_2$ or $a_2-1\le_2 b_2$. This gives at most $2\cdot 3^{\ell-\ell_s}$ choices for the pair $(a_2,b_2)$.
  Given $a_2$ and $b_2$, there are $2^{\ell_s}$ choices for each of $a_1$ and $b_1$, for a total of at most $2\cdot 3^{\ell-\ell_s}\cdot 4^{\ell_s}$ solutions to \eqref{eq:mult-5}.
\end{proof}
\begin{lemma}
  \label{lem:mult-5}
  Let $r=2^{\ell_r}$, $s=2^{\ell_s}$ and $q=2^\ell$ with $\ell_r,\ell_s\in\{1,\dots,\ell-1\}$.
  The number of $(q,r,rq-s)$-good monomials is at least $\binom{r+1}{2}4^\ell-3rs^{\log_2(4/3)} \cdot 3^\ell$.  
\end{lemma}
\begin{proof}
  The number of type-$r$ monomials is $\binom{r+1}{2}q^2 = \binom{r+1}{2}4^\ell$.
  A monomial $M_{a,b}$ is $(q,r,rq-s)$-good if, for every $\alpha,\beta\in\mathbb{F}_q$, we have
  \begin{align}
    M_{a,b,\alpha,\beta}(T) \ \defeq \ T^a(\alpha T+\beta)^b = \sum_{i=0}^{b} \alpha^i \beta^{b-i} T^{a+i}\binom{b}{i}.
  \label{}
  \end{align}
  can be represented as a polynomial of degree less than $rq-s$.
Next, we apply Lemma~\ref{lem:mult-2}, which says that there is a unique polynomial $B(T)$ so that $\deg(B) \leq rq - 1$ so that $B(T) \equiv_r M_{a,b,\alpha,\beta}(T)$, and further that all of the coefficients $[T^c]B(T)$ for $\deg(M_{a,b,\alpha,\beta}) - (qr -r ) < c < qr$ are equal to the corresponding coefficient of $B(T)$.  
As $M_{a,b}$ is type $r$, we have $\floor{a/q}+\floor{b/q}<r$, so the degree of the polynomial $M_{a,b,\alpha,\beta}$ is at most $a+b\le (r+1)q-2$, and 
\[ ((r+1)q - 2) - (qr -r ) = r + q -2 < qr - s \]
for any allowed choice of $q,r,s$, so $[T^c]B(T) = [T^c] M_{a,b,\alpha,\beta}(T)$ for all $c$ so that
\[ qr - s \leq c \leq qr. \]
Thus, to show that $B(T)$ has degree less than $qr - s$, it suffices to show that 
the coefficients of $T^{qr-s}, T^{qr-s+1},\dots,T^{qr-1}$ in $M_{a,b,\alpha,\beta}$ are all zero.

  Write $a=a_0q+a_1$ and $b=b_0q+b_1$ where $a_0+b_0\le r-1$ and $0\le a_1,b_1\le q-1$.
  Note that if $a_0+b_0 < r-1$, then for $s'=1,\dots,s$ coefficient $[T^{rq-s'}]M_{a,b,\alpha,\beta}$ is always zero except possibly when $a_0+b_0=r-2$ and $a_1+b_1\ge 2q-s$.
  This can happen for at most $\frac{rs^2}{2}$ pairs $(a,b)$.
  Hence, for $a_0+b_0<r-1$, there are $\le \frac{rs^2}{2}$ bad monomials $(a,b)$.

  Now assume $a_0+b_0=r-1$.
  For $s'=1,\dots,s$, the coefficient of $T^{rq-s'}$ in $T^{a}(\alpha T+\beta)^{b}$ is 0 if $rq-s' < a$ or $a+b < rq-s'$.
  Otherwise, the coefficient is
  \begin{align}
    \alpha^{rq-s'-a}\beta^{b-rq+s'+a}\binom{b}{rq - s' - a}
    \ &= \  \alpha^{rq-s'-a}\beta^{b-rq+s'+a}\binom{b_0q+b_1}{b_0q+q-s'-a_1}.
  \label{}
  \end{align}
  By Proposition~\ref{prop:lucas}, the binomial coefficient is nonzero (mod 2) if and only if $b_0q + q-s'-a_1\le_2 b_0q+b_1$, which, as $q$ is a power of 2, happens only if $q-s'-a_1\le_2^\ell b_1$.
  Hence, if $a_0+b_0=r-1$, the monomial $M_{a,b}$ is $(r,rq-s)$-bad only if some $s'=1,\dots,s$ satisfies $q-s'-a_1\le_2^\ell b_1$.
  Hence, by Lemma~\ref{lem:mult-3}, for a fixed $a_0,b_0$ with $a_0+b_0=r-1$, there are at most $2s^{\log_2(4/3)}3^\ell$ bad monomials $M_{a,b}$, so there are at most $r\cdot s^{\log_2(4/3)}3^\ell$ bad monomials $M_{a,b}$ over all $a_0,b_0$ with $a_0+b_0=r-1$.
  As we showed, there are at most $\frac{rs^2}{2}$ bad monomials when $a_0+b_0<r-1$.
  Hence, there are at least $\binom{r+1}{2}4^\ell-2rs^{\log_2(4/3)}3^\ell-\frac{rs^2}{2} \ge \binom{r+1}{2}q^2 - 3rs^{\log_2(4/3)}q^{\log_2(3)}$ good monomials, as desired.
\end{proof}
Lemma~\ref{lem:mult-5} and Lemma~\ref{lem:good} together imply Corollary~\ref{cor:mult-6}, which in turn implies the informal result stated at the beginning of the section.

\begin{corollary}
  \label{cor:mult-6}
  Let $r=2^{\ell_r}$, $s=2^{\ell_s}$ and $q=2^\ell$ with $\ell_r,\ell_s\in\{1,\dots,\ell-1\}$.
  A $(q,r,rq-s)$ lifted multiplicity code has rate at least $1-6r^{-1}s^{\log_2(4/3)}q^{\log_2(3/4)}$.
\end{corollary}

\begin{remark}
  We apply Corollary~\ref{cor:mult-6} for $r=s\le q$, giving that a lifted multiplicity code of rate at least $1-6r^{\log_2(2/3)}q^{\log_2(3/4)}$.
  By comparison \cite{KSY14}, a 2-variate multiplicity code of order $r$ evaluations of degree at most $rq-r$ polynomials over $\mathbb{F}_q$ has rate $\frac{\binom{rq-r+2}{2}}{\binom{r+1}{2}q^2} \le 1-\Omega(\frac{1}{r})$, which is smaller than the rate of lifted multiplicity codes for $r\ll q$.
\end{remark}
\begin{remark}[Quantitative comparison to \cite{Wu15}]\label{rem:Wu15}
The work \cite{Wu15} also studies lifted multiplicity codes, but focuses on a different parameter regime than we focus on here (where $t$ is large, rather than $t \leq \sqrt{N}$).  Perhaps because they focus on a different parameter regime, the approach of \cite{Wu15} does not yield any nontrivial results in our parameter regime, and consequently our analysis of lifted multiplicity codes is much stronger. 

  For example, for degree $d=rq-r$ codes, \cite{Wu15} bounds%
\footnote{
    The details are as follows: using some notation from \cite{Wu15}, for a rate $1-\alpha$ code in our parameter regime, $n=2$ variables, a prime $p=2$, a parameter $b=\ceil{\log_p n}+1 = 2$, $q=p^\ell$, $r=p^{\ell_r}$ (they use $m$ for $r$, $s$ for $\ell$, and $t$ for $\ell_r$), $\alpha_t = \frac{\alpha}{1 - \frac{N_n(p^{{\ell_r}-2})}{N_n(p^{\ell_r})}} = \frac{\alpha}{1-\frac{N_2(r/4)}{N_2(r)}} \approx \frac{\alpha}{1-1/16} = \Theta(\alpha)$ (here $N_2(r) \defeq \binom{r+1}{2}$, and we assume $r\gg 1$), $c = bp^{bn}\ln\frac{1}{\alpha_t} +{\ell_r} = 32\ln\frac{1}{\alpha_t}+\ell_r$, $d=(1-p^{-c})rq$.
    In our setting, we choose $d=rq-r$, which requires $p^{c} = q$, so $c=\ell$, so $\ell-\ell_r = 32\ln\frac{1}{\alpha_t} = 32\ln 2\cdot \log\frac{1}{\alpha} -O(1)$.
    Thus, $\alpha = \Theta(2^{(\ell_r-\ell)/(32\ln 2)}) = (\frac{r}{q})^{1/(32\ln2)}$.
  }
the rate of the code below by $1-\Theta(\frac{r}{q})^{1/(32\ln2)}$.
This is a weaker bound than the straightforward bound of $1 - q^{\log_2(3/4)}$ sketched at the beginning of this section, and significantly weaker than our bound in Corollary~\ref{cor:mult-6} of $1 - 6r^{\log_2(2/3)}q^{\log_2(3/4)}$ for all $r$ and $q$.  Moreover, for $r \gg q^{1/23}$, the bound of \cite{Wu15} is even weaker than the lower bound on the rate of (non-lifted) multiplicity codes, which is $1 - \Omega(1/r)$.
\end{remark}

\begin{remark}[The value of bivariate lifts]\label{rem:bivariate}
In addition to likely giving better bounds than $m$-variate lifts (see Why only bivariate lifts? in Section~\ref{sec:related}), another reason that we study only bivariate lifts in this paper is that it makes the computations much more tractable.  In the proof of Lemma~\ref{lem:mult-5}, we study $M_{a,b,\alpha,\beta}(T) = (T^a)(\alpha T + \beta)^b$, and expand out the terms to apply Lucas's theorem.  If we were to consider, say, trivariate lifts, we would have to expand expressions of the form $(T^a)(\alpha T + \beta)^b(\gamma T + \delta)^c$, and it would become more complicated to keep track of the coefficients on various powers.  Analyzing $m$-variate lifts would become more complicated still.  In particular, it seems harder to get as tight a bound on the codimension of the code for $m$-variate lifts for $m > 2$ as we are able to get for $m=2$.  Given that we are already able to obtain good codes for bivariate lifts, we restrict our attention to this simpler case.
\end{remark}

\section{Disjoint repair groups of lifted multiplicity codes}\label{sec:drgp}
Finally, we prove Theorem~\ref{thm:main}, which we repeat below.
\begin{theorem*}[Theorem~\ref{thm:main}, restated]
  Let $r=2^{\ell_r}$ and $q=2^\ell$ with $\ell_r < \ell$ and $\mathcal{C}$ be the $(q,r,rq-r)$ lifted multiplicity code.
  \begin{itemize}
  \item The length of the code is $q^2$.
  \item The rate of the code is at least $1 - 6r^{\log_2(2/3)}q^{\log_2(3/4)}$.
  \item The code has the $q/r$-disjoint repair group property.
  \end{itemize}
\end{theorem*}
\begin{proof}
  The first item follows from the definition of $\mathcal{C}$, and the second item is by Corollary~\ref{cor:mult-6}.
  To see the third item, we observe that, given a point $(\gamma,\delta)\in\mathbb{F}_q^2$, lines $L_1,\dots,L_r$ passing through $(\gamma,\delta)$, and $P^{(<r)}(\mathbf{y})$ at all points $\mathbf{y}$ on the lines $L_1,\dots,L_r$ except $(\gamma,\delta)$ itself, we can (efficiently) recover 
$P^{(<r)}(\gamma, \delta)$.
  This guarantees the $q/r$-disjoint repair group property, because we can group the $q$ lines of $\mathcal{L}$ of the form $L(T) = (T, \alpha T + \beta)$ passing through $(\gamma,\delta)$ arbitrarily into groups of $r$, giving $q/r$ disjoint repair groups.
  For any line $L_k$, the polynomial $P_{L_k}(T)$ has degree at most $rq-r-1$, as $P$ is $(q,r,qr-r)$-good.
  By taking linear combinations of directional derivatives (Lemma~\ref{lem:local-1}), we can efficiently compute $P\ind{i}_{L_k}(\gamma')$ for every $i=0,\dots,r-1$, every $k=1,\dots,r$, and every $\gamma'\neq \gamma$.
  We can compute $P_{L_k}(T)$ using a generalization of polynomial interpolation.
  This can be done in $O(D\log D)$ time, where $D<rq$ is the degree of the polynomial (see e.g. \cite{Chin76})
  Hence, by Corollary~\ref{lem:local-2}, from $P_{L_1}(T),\dots,P_{L_k}(T)$, we can efficiently compute $P\ind{i,j}(\gamma,\delta)$ for all $i,j$ with $0\le i+j\le r-1$.
\end{proof}

\section{Conclusion}\label{sec:conclude}
We conclude with some open questions.
\begin{enumerate}
\item We have shown that lifted multiplicity codes with redundancy $O(t^{0.585} \sqrt{N})$ have the $t$-DRGP for a range of $t \leq \sqrt{N}$.  However, we do not know of any general lower bounds when $t\in(1,\sqrt{N})$ beyond the lower bound for $t=2$, which implies that the redundancy must be at least $\Omega(\sqrt{N})$ for any $t$.  
When $t\ge \sqrt{N}$, there is a stronger redundancy lower bound of $\Omega(t)$, which holds simply because a code with the $t$-DRGP must have Hamming distance at least $t$.
Thus, it is an open question whether or not our bound is tight or whether one can do better.
\item Lifted multiplicity codes display better locality for the $t$-DRGP problem for $t \leq \sqrt{N}$; it is a natural question to ask whether they can be used for larger $t$, and in particular whether they could lead to improved constructions of locally correctable codes. 
In particular, it would be interesting if lifted multiplicity codes could qualitatively out-perform (un-lifted) multiplicity codes as high-rate LCCs, for example by maintaining the high rate while achieving sub-polynomial query complexity.\footnote{As noted in the introduction, 
the work \cite{Wu15} showed the lifted multiplicity codes are good LCCs with lower-order derivatives than were required by the (un-lifted) multiplicity codes of \cite{KSY14}, but it does not show how to improve the query complexity to sub-polynomial.}
We note that for the LCC problem, one typically does not care about pinning down the rate, so long as it is close to $1$, instead focusing on the query complexity.  In contrast, in this work, we have focused on pinning down the rate much more precisely.

\item Related to the above, it would be natural to understand the rate and locality of lifted multiplicity codes over more than two variables.
\item The alphabet size of lifted multiplicity codes is $q^{\binom{r+1}{2}}$, which, if the multiplicity is $r=q^\alpha$ for a constant $\alpha>0$, is exponential type in the code length $q^2$. In practical applications, a smaller alphabet size is desirable. It would be interesting to achieve the results of Corollary~\ref{cor:main} with a code whose length grows independently of the alphabet size. 
\item 
In this paper, we studied the $t$-DRGP locality property, which requires  that each symbol has many disjoint repair groups. Another common notion of locality is an Locally Recoverable Code (LRC) with locality $d$, which requires that each symbol has one repair group of size at most $d$.   These two notions are combined in the notion of an LRC with locality $d$ and availability $t$ (see, e.g. \cite{TB_mult_rec, WZ14, RPDV14}), combines these two notions.  This requires that each symbol have $t$ disjoint repair groups, each of size at most $d$. The techniques in this paper yield codes with locality $rq$ and availability $q/r$, where $r$ is the multiplicity. It would be interesting to construct codes with a better trade-off between locality and availability, possibly using lifting and/or multiplicity techniques.
\end{enumerate}

\section*{Acknowledgements} 
We thank Eitan Yaakobi for helpful conversations.
We thank Julien Lavauzelle for pointing out the reference \cite{Wu15}, for pointing out an error in an earlier version of this paper, and for suggesting the fourth open question.
We thank Nikita Polianskii for pointing out an error in an earlier version of this paper.
A previous version claimed that a lifted code is exactly the span of all good monomials, but in fact the span of all good monomials only forms a subset of the lifted code (see Remark~\ref{rem:err}).
This does not change our main result, as our lower bound on the number of good monomials still gives the same lower bound on the rate of the lifted code.
We thank anonymous reviewers for helpful comments on an earlier draft of this paper.

\bibliographystyle{alpha}
\bibliography{bib}

\appendix

\section{Proofs of polynomial facts}
\label{app:A}

\begin{proof}[Proof of Proposition~\ref{prop:deriv-2}]
  By part 2 of Proposition~\ref{prop:deriv-1},
  \begin{align}
    P\ind{i}(X) = \sum_{j_1+\cdots+j_r=i}^{} \prod_{k=1}^{r} D\ind{j_k}(X^q-X).
  \end{align}
  We have $D\ind{1}(X^q-X)=1$ (the field has characteristic 2).
  For $2\le i < q$, the $i$th derivative of $X^q-X$ is $\binom{q}{i}X^{q-i}$, which is 0, as $\binom{q}{i}$ is even by Proposition~\ref{prop:lucas}.
  The summand above is nonzero if and only if $j_1,j_2,\dots,j_r\le1$.
  When $i\le r$, this happens when $i$ of the $j_k$'s are 1 and $r-i$ are 0, which happens for $\binom{r}{i}$ choices of $(j_1,\dots,j_r)$.
  This gives $P\ind{i}(X)=\binom{r}{i}(X^q-X)^{r-i}$ for $0\le i\le r$.
  When $i>r$, some $j_k$ is at least 2, in which case $P\ind{r}(X)=0$ for $r<i<q$.
\end{proof}

\begin{proof}[Proof of Lemma~\ref{lem:local-1}]
  Let $\bfa_k$ denote the vector $(1,\alpha_k)$, and let $\bfb_k$ denote the vector $(0,\beta_k)$. 
  By assumption, we have that $\bfa_k\gamma + \bfb_k = (\gamma,\delta)$.
  By the definition of Hasse derivatives, we have, for all $k=1,\dots,r$
  \begin{align}
    P_{L_k}(T + Z)
    \ &= \  P(\bfa_k T + \bfb_k + \bfa_k Z) \nonumber\\
    \ &= \  \sum_{\bfi\in\mathbb{N}^2}^{} P\ind{\bfi}(\bfa_k T + \bfb_k)\cdot  (\bfa_k Z)^{\bfi} \nonumber\\
    \ &= \  \sum_{\bfi\in\mathbb{N}^2}^{} P\ind{\bfi}(\bfa_k T + \bfb_k) \cdot \bfa_k^\bfi Z^{\wt(\bfi)}\nonumber\\
    P_{L_k}(T+Z)
    \ &= \ \sum_{i\ge 0}^{} P\ind{i}_{L_k}(T) Z^i
  \label{}
  \end{align}
  Hence, for all $i\ge 0$ and $k=1,\dots,r$, we have
  \begin{align}
    P_{L_k}\ind{i}(T)  
    \ &= \ \sum_{\bfi:\wt(\bfi) = i}^{} P\ind{\bfi}(\bfa_k T + \bfb_k) \bfa_k^\bfi
  \label{}
  \end{align}
  By plugging in $T=\gamma$, we have for all $i\ge 0$ and $k=1,\dots,r$,
  \begin{align}
    P_{L_k}\ind{i}(\gamma)
     \ &= \ \sum_{\bfi:\wt(\bfi) = i}^{} P\ind{\bfi}(\gamma,\delta) \bfa_k^\bfi.
  \label{}
  \end{align} 
  Rewriting this in matrix form gives the desired result.
\end{proof}

\section{Lifted codes via dual codes}\label{app:dual}
It was shown in \cite{GuoKS13} that bivariate lifted parity-check codes over $\mathbb{F}_q$, where $q=2^\ell$, have co-dimension $3^\ell$.
Here, we give an alternative proof using dual codes.
The techniques in this proof are not directly related to the techniques that we used in the main body of the paper, but we found this alternative proof illuminating so we include it.

Let $q=2^\ell$.
Recall $\mathcal{L}$ is the set of lines expressible as $L(T)=(T,\alpha T + \beta)$ where $\alpha,\beta\in \mathbb{F}_q$.
One way to think about codes with locality is by considering their dual code.
If the code is a subset of $\mathbb{F}_q^{q\times q}$, then the dual code corresponds to lines of repair groups.
Given a line $L(T)$ in $\mathcal{L}$, define the corresponding dual codeword:
\begin{align}
  (c^\perp_L)_{ij} \defeq \left\{
  \begin{tabular}{ll}
  1 & $(i,j)=L(t)$ for some $t\in\mathbb{F}_q$\\
  0 & \text{o/w}\\
  \end{tabular}
  \right.
\end{align}
Let 
\begin{align}
  V_\mathcal{L} \ &\defeq \ \spn\left\{c^\perp_L:L\in \mathcal{L}\right\}.
\label{}
\end{align}
Note that $V_{\mathcal{L}}$ is spanned by $4^\ell$ elements, so the trivial bound on the dimension is $4^\ell$.
We give the following improved bound, matching the analysis of \cite{GuoKS13}.
\begin{lemma}
  \label{lem:dual-1}
  The subspace $V_{\mathcal{L}}$ has dimension at most $3^\ell$.
\end{lemma}
\begin{proof}
  A codeword $c^\perp_L$ is the evaluation of the following polynomial on $\mathbb{F}_q^{q\times q}$:
  \begin{align}
    P_L(X,Y)\ &\defeq \  \prod_{\beta\neq \beta_L}^{} (\alpha_L X + \beta - Y).
  \label{}
  \end{align}
  If $(X,Y)\notin L$, then the polynomial evaluates to 0 as $Y-\alpha_L X \neq \beta_L$, and otherwise it evaluates to
  \begin{align}
    \prod_{\beta\neq\beta_L}^{} (\beta-\beta_L) = \prod_{\beta\in\mathbb{F}_q^*}^{} \beta = 1.
  \end{align}
  For $a+b\ge q$, the coefficient of $X^aY^b$ in $P_L(X,Y)$ is 0.
  For $a+b\le q$, the coefficient of $X^aY^b$ in $P_L(X,Y)$ is
  \begin{align}
    \binom{a+b}{a}\alpha_L^a(-1)^b\sum_{\substack{\beta_1,\dots,\beta_{q-1-a-b}\in\mathbb{F}_q\\ \text{distinct}, \neq \beta_L}}^{} \prod_{j=1}^{q-1-a-b} \beta_j
    .
  \label{eq:coeffab}
  \end{align}
  This is because we first chose $a+b$ terms that contain $X$ or $Y$, then choose which terms are $X$ and which terms are $Y$, and this gives us $a$ many $\alpha_L$'s and $b$ many $-1$'s, and we sum over the choices of the $\beta$ terms that we choose.
  Hence, the only $a,b$ such that $[X^aY^b]P_L(X,Y)\neq0$ for any $L$ are the pairs $(a,b)$ such that $a+b\le q-1$ and $\binom{a+b}{a}\equiv 1\mod 2$.
  There are at most $3^\ell$ pairs by Proposition~\ref{prop:lucas}.
  It follows that the polynomials $P_L(X,Y)$ are spanned by $3^\ell$ monomials $X^aY^b$ with $\binom{a+b}{a}\equiv 1\mod 2$.
  Hence, the vector space $V_{\mathcal{L}}$ is spanned by $3^\ell$ dual codewords in $\mathbb{F}_q^{q\times q}$ and thus has dimension at most $3^\ell$.
\end{proof}
\end{document}